\documentclass[11pt]{article}
\usepackage[utf8]{inputenc} 
\usepackage[normalem]{ulem}
\usepackage{amsthm}
\usepackage{mathtools}
\usepackage{authblk}
\usepackage{thmtools}

\usepackage{latexsym}
\usepackage{amsmath}
\usepackage{amssymb}
\usepackage{mathtools}
\usepackage{amsthm}
\usepackage{epsfig}

\usepackage{layout}
\usepackage{manfnt}
\usepackage{datetime}
\usepackage{esvect}
\usepackage{fullpage}
\usepackage{subfigure}
\usepackage{fancyheadings}
\usepackage{enumerate}
\usepackage{epstopdf, bm}
\usepackage{epsfig,subfigure}
\usepackage{graphicx}

\newtheorem{theorem}{Theorem}[section]
\newtheorem{corollary}{Corollary}[section]

\usepackage{biblatex}
\declaretheoremstyle[headfont=\normalfont]{normalhead}
\usepackage{color}

\usepackage{graphicx}
\usepackage{morefloats}
\usepackage{hyperref}
\usepackage{ amssymb }
\usepackage{textcomp}

\usepackage{url}

\newcommand{\filt}{\ensuremath{\mathcal{F}}}

\usepackage[margin=1in]{geometry}
\usepackage{float}

\usepackage{etoolbox}

\begin{document}


\makeatletter
\patchcmd{\maketitle}{\@fnsymbol}{\@alph}{}{}  
\makeatother

\title{Pricing Exchange Rate Options and Quanto Caps in the Cross-Currency Random Field LIBOR Market Model}
\author{
   Rajinda Wickrama\thanks{Department of Mathematics, University of Iowa, Iowa City, Iowa. 
   Email: rajinda-wickrama@uiowa.edu}

}
\date{}

\date{}

\maketitle
\begin{abstract}
\hskip -.2in
\noindent

We develop an arbitrage-free random field LIBOR market model to price cross-currency derivatives. The uncertainty of the forward LIBOR rates of our cross-currency model is driven by a two time parameter random field instead of a finite dimensional Brownian motion. To demonstrate the applications of this model, we develop an approximate closed-form pricing formula for Quanto caps and cross-currency swaps. Further, we derive an exact pricing formula for an exchange rate option in the random field setting. 
\end{abstract}
\section{Introduction}
The London Inter-bank offered Rates (LIBOR) are reference interest rates at which banks lend loans to other banks in London \cite{lixinwu2019interest}. Many authors have worked and developed closed-form derivative prices based on LIBOR rates by modeling uncertainties with a finite dimensional Brownian motion (for example see \cite{brace1997market,jamshidian1997libor,jamshidian1999libor,musiela1997continuous}). Despite having many advantages, there are very few articles based on pricing interest-rate derivatives using random field interest models and the RFLMM. The main contribution of this paper is the construction of the cross-currency RFLMM. The cross-currency LIBOR market models were first developed by Mikkelsen \cite{mikkelsen2001cross} and Schl\"ogl\cite{schlogl} while the RFLLM was developed by Wu and Xu \cite{wu}. The uncertainty in the model developed by Schl\"ogl is driven by a standard Brownian motion. Scholg's model was extended to the multi-factor setting by Amin \cite{amin2003multi} and Brenner et al. \cite{benner2009multi}. Beveridge et al.\cite{beveridge2010efficient} develop a discussion on displaced diffusion LIBOR market model in the cross-currency setting to price Bermudan options. The main result by Schl\"ogl \cite{schlogl} is that it is not possible to assume all forward LIBOR rates of all economies and the forward exchange rates are lognormal simultaneously. We will show that this also holds true in the random field setting.\\

The structure of the article is as follows. In section 2, we briefly discuss the random field term structure model developed by Goldstein\cite{goldstein2000term}, and the random field LIBOR market model (RFLMM) developed by Wu and Xu \cite{wu}. Our main goal is to apply these two frameworks to extend the cross-currency LMM \cite{schlogl} to the random field setting. We construct the cross-currency RFLMM in section 3 and prove that even in the random field setting, it is not possible to assume that all the volatility terms for all maturities of the forward exchange rates are deterministic when all forward LIBOR rates are assumed to be lognormal. In section 4, we find an approximate valuation formula for the price of a Quanto cap in terms of the domestic currency. In section 5, we construct an approximate pricing formula a float-to-float cross currency swap and in section 5 we derive an exact formula for a call option written on the spot exchange rate. A similar formula in the finite dimensional Brownian motion case was derived by Musiela et.al. \cite{musiela}.

\section{Random field forward interest models}

We begin this section by stating the dynamics of the domestic and the foreign instantaneous forward rates with respect to the the domestic and the foreign risk-neutral measures $Q$ and $Q^F$. Working in the Goldstein framework \cite{goldstein2000term} we state the following: 

\begin{equation*}
    df(t,T)=\mu^Q(t,T)dt+\sigma(t,T)dZ^Q(t,T),
\end{equation*}

\begin{equation*}
    df_F(t,T)=\mu_F^{Q^F}(t,T)dt+\sigma_F(t,T)dZ^{Q^F}(t,T).
\end{equation*}
where $Z^{Q}$ and $Z^{Q^F}$ are random fields satisfying the following conditions (see \cite{santa2001dynamics}): 

\begin{enumerate}
    \item $Z^P(t,T)$ is continuous for all $t,T.$
    \item For any fixed $T$, the process $\{Z^P(t,T)\}_{t\leq T}$ is a $P-$martingale such that $E^P[dZ^P(t,T)]=0$ and $var^P[dZ(t,T)]=dt$ for all $t\leq T.$ 
    \item $dZ^P(t,T_1)dZ^P(t,T_2)=c(T_1,T_2)dt$ where $c$ is the correlation function of the random field which satisfies $C(T,T)=1.$
\end{enumerate}
We assume that the above assumptions are true for $P=Q,Q^F$ and that $c$ is a deterministic function. We also make the additional assumption that $dZ^Q(t,T_1)dZ^{Q^F}(t,T_2)=c(T_1,T_2)dt.$ The filtration generated by the random fields is defined by $\filt_t=\sigma\{Z(u,v)\mid\ u\leq t, u\leq v\}$.\\

    For a fixed $T$, the process $Z^P(t,T)$ is a continuous martingale with quadratic variation $d<Z^P(t,T),Z^P(t,T)>=dt.$ By Levy's characterization theorem (see theorem 3.16 in \cite{karatzas}), $Z^P(t,T)$ is a Brownian motion (with respect to $P$) for all $T\geq 0$. Goldstein \cite{goldstein2000term} proves that the no-arbitrage drift $\mu^P(t,T)=\sigma(t,T)\int_t^T \sigma(t,v)c(T,v)dv.$ The above model is a generalization of the Heath-Jarrow-Morton (HJM) Model \cite{heath1992bond} to a random field. The HJM model, which is a framework to model instantaneous forward rates, was  constructed with the volatility term driven by a finite (but arbitrary) number of independent Brownian motions. This was first generalized to a Gaussian random field setting by Kennedy \cite{kennedy1994term} and then by Goldstein \cite{goldstein2000term} to non-Gaussian random fields. \\

By $B(t,T)$ and $B_F(t,T)$ we denote the prices of the domestic and foreign zero-coupon bonds. These prices are calculated by 

\begin{equation*}
    B(t,T)=\exp\left(-\int_t^T f(t,u)du\right),
\end{equation*}
\begin{equation*}
    B_F(t,T)=\exp\left(-\int_t^T f_F(t,u)du\right).
\end{equation*}

By a straightforward application of Ito's lemma (\cite{oks} theorem 4.1.2) we can derive the dynamics of the bond prices as

\begin{equation}\label{bonddynamics}
    \frac{dB(t,T)}{B(t,T)}=r(t)dt-\int_t^T \sigma(t,u)dZ^Q(t,u)du,
\end{equation}

\begin{equation}\label{fbonddynamics}
    \frac{dB_F(t,T)}{B_F(t,T)}=r_F(t)dt-\int_t^T \sigma_F(t,u)dZ^{Q^{F}}(t,u)du
\end{equation}

where $r(t)=f(t,t)$ and $r_F(t)=f_F(t,t)$. Goldstein \cite{goldstein2000term} identifies the $T-$ forward measure $Q_T$ (by taking $B(t,T)$ to the the numeraire asset) by the relationships

\begin{equation}\label{dzt}
    dZ^{Q_T}(t,u)=dZ^Q(t,u)+\left(\int_t^T\sigma(t,v)c(u,v)\right)dt,
\end{equation}

\begin{equation}\label{dzft}
    dZ^{Q^F_T}(t,u)=dZ^{Q^F}(t,u)+\left(\int_t^T\sigma_F(t,v)c(u,v)\right)dt
\end{equation}

where $Z^{Q_T}(t,T)$ and $Z^{Q^F_T}(t,T)$ are $Q_T$ and $Q_T^F$ random fields respectively. Therefore, the price of a derivative with expiration date $T$ at time $t$ in terms of the forward measure $Q_T$ is:

\begin{equation}\label{forward}
    V(t)=B(t,T)E^{Q_T}_t[V(T)].
\end{equation}
\eqref{forward} holds true because $V(t)/B(t,T)$ is a $Q_T$-martingale. By the definition of a martingale, $$V(t)/B(t,T)=E^{Q_T}[V(T)/B(T,T)\mid \filt_t]$$
for all $t\leq T.$ We now illustrate how it can be proved that the relative bond price $B(t,U)/B(t,T)$ is a $Q_T-$martingale.



\begin{theorem}\label{relativebond}
The dynamics of $B(t,U)/B(t,T)$ with respect to $Q_T$ is given by 

\begin{equation}\label{bondqt}
    d\left(\frac{B(t,U)}{B(t,T)}\right)=\frac{B(t,U)}{B(t,T)}\left[\int_U^T \sigma(t,u)dZ^{Q_T}(t,u)du\right].
\end{equation}

\end{theorem}

\begin{proof}
Observe that
\begin{equation*}
    d\left(\frac{1}{B(t,T)}\right)=\frac{1}{B(t,T)}\left[-r(t)dt+\int_t^T\sigma(t,u)dZ^Q(t,u)du+\int_t^T\int_t^T\sigma(t,u)\sigma(t,v)c(u,v)dudvdt\right].
\end{equation*}

Applying Ito's product rule to $B(t,U)/B(t,T)$ yields the following:

\begin{equation*}
\begin{split}
     d\left(\frac{B(t,U)}{B(t,T)}\right)&=B(t,U)d\left(\frac{1}{B(t,T)}\right)+\frac{1}{B(t,T)}dB(t,U)+d\left<\frac{1}{B(t,T)},B(t,U)\right>\\
     &=\left(\frac{B(t,U)}{B(t,T)}\right)\Bigg[r(t)dt-\int_t^U\sigma(t,u)dZ^Q(t,u)du-r(t)dt+\int_t^T\sigma(t,u)dZ^Q(t,u)du \\
     &+\int_t^T\int_t^T\sigma(t,u)\sigma(t,v)c(u,v)dudvdt-\int_t^T\int_t^U\sigma(t,u)\sigma(t,v)c(u,v)dudvdt\Bigg].
\end{split}
\end{equation*}

Applying \eqref{dzt} to the above equation gives the desired result.
\end{proof}

\section{The Cross-Currency RFLMM}

Fix a tenor $0\leq T_1<T_2<...<T_N\leq T$. Schl\"ogl \cite{schlogl} defines a process $X(t)$ as the spot exchange rate at time $t$ in terms of the domestic currency per one unit of foreign currency. $X(t)$ is assumed to be strictly positive martingale with respect to the $T-$ forward measure $Q_T$.  The $T_i-$forward exchange rate is defined as

\begin{equation}
    X(t,T_i):=\frac{B_F(t,T_i) X(t)}{B(t,T_i)}
\end{equation}

where $B_F(t,T_i)$ here denotes the price of a bond quoted in some foreign currency. $X(t,T_i)$ can be interpreted as the spot exchange rate at a future date $T_i$ as seen from time $t$. For the purposes of this paper, we will not be directly specifying the dynamics of the spot exchange rate $X(t)$ and instead we will work with the dynamics of the forward exchange rate. The dynamics of spot exchange rate $X(t)$ in the finite dimensional Brownian motion case can be found in \cite{musiela}, \cite{brace2007engineering}.\\

We assume that the dynamics of $X(t,T_i)$ with respect to $Q_{T_i}$ is given by

\begin{equation}
    dX(t,T_i)=X(t,T_i)\int_t^{T_i}\sigma_{X_{i}}(t,u)dZ^{Q_{T_i}}(t,u)du. 
\end{equation}

for all $i=1,2,...,N$ and $\sigma_{X_i}$ is the volatility of $X(t,T_i)$. By  applying Ito's lemma to $\ln X(t,T_i)$ we can see that

\begin{equation}
\begin{split}
     X(t,T_i)=X(0,T_i)\exp\Bigg(&\int_0^t\int_s^{T_i}\sigma_{X_{i}}(s,u)dZ^{Q_{T_i}}(s,u)du\\&-\frac{1}{2}\int_0^t\int_s^{T_i}\int_s^{T_i}\sigma_{X_{i}}(s,u)\sigma_{X_{i}}(s,v)c(u,v)dudvds\Bigg)   
\end{split}
\end{equation}
for all $t\leq T_i.$ Notice that 

\begin{equation*}
    \eta_i(t):= \exp\Bigg(\int_0^t\int_s^{T_i}\sigma_{X_{i}}(s,u)dZ^{Q_{T_i}}(s,u)du-\frac{1}{2}\int_0^t\int_s^{T_i}\int_s^{T_i}\sigma_{X_{i}}(s,u)\sigma_{X_{i}}(s,v)c(u,v)dudvds\Bigg) 
\end{equation*}

is a $Q_{T_{i}}$-martingale and $E^{Q_{T_{i}}}[\eta_i(T_i)]=1.$

\subsection{The dynamics of the foreign forward LIBOR rate with respect to the domestic forward measure.}

Now we state some key results developed Wu and Xu \cite{wu}
and extend it to the cross-currency format. Fix a discrete tenor $0< T_0<T_1<...<T_N\leq T.$ The domestic forward LIBOR rate $L(t,T_i)$ is defined in terms of bond prices in the following way:

\begin{equation*}
    1+\delta_{i+1}L(t,T_i)=\frac{B(t,T_{i})}{B(t,T_{i+1})}
\end{equation*}

where $\delta_{i+1}=T_{i+1}-T_i.$
Similarly, the foreign forward LIBOR rate $L_F(t,T_i)$ is defined by 

\begin{equation*}
    1+\delta_{i+1}L_F(t,T_i)=\frac{B_F(t,T_{i})}{B_F(t,T_{i+1})}.
\end{equation*}
As a consequence of theorem \eqref{relativebond}, $L(t,T_i)$ and $L_F(t,T_i)$ are martingales with respect to $Q_{T_{i+1}}$ and $Q^F_{T_{i+1}}$ respectively. 
Using \eqref{dzt} and \eqref{dzft}, we can see that $Q_{T_{i+1}}$, $Q_{T_i}$  and $Q^F_{T_{i+1}}$, $Q^F_{T_i}$ are related by

\begin{equation}\label{dzq}
    dZ^{Q_{T_{i+1}}}(t,u)=dZ^{Q_{T_i}}(t,u)+\left(\int_{T_i}^{T_{i+1}}\sigma(t,v)c(u,v)dv\right)dt
\end{equation}

and

\begin{equation}\label{dzqf}
    dZ^{Q^F_{T_{i+1}}}(t,u)=dZ^{Q^F_{T_i}}(t,u)+\left(\int_{T_i}^{T_{i+1}}\sigma_F(t,v)c(u,v)dv\right)dt.
\end{equation}

Applying Ito's rule to the definitions of the domestic and foreign LIBOR rates together with the equations \eqref{dzq} and \eqref{dzqf}, it can be shown that

\begin{equation}
    dL(t,T_i)=\frac{B(t,T_{i})}{\delta_{i+1} B(t,T_{i+1})}\int_{T_i}^{T_{i+1}}\sigma(t,u)dZ^{Q_{T_{i+1}}}(t,u)du
\end{equation}

and 

\begin{equation}
    dL_F(t,T_i)=\frac{B_F(t,T_{i})}{\delta_{i+1} B_F(t,T_{i+1})}\int_{T_i}^{T_{i+1}}\sigma_F(t,u)dZ^{Q^F_{T_{i+1}}}(t,u)du.
\end{equation}

Since $L(t,T_i)$ and $L_F(t,T_{i})$ are $Q_{T_{i+1}}$ and $Q^F_{T_{i+1}}$ martingales respectively, by the martingale representation theorem (\cite{oks} theorem 4.3.4), there exist $\filt_t$ adapted functions $\xi_i$ and $\xi_i^F$ such that 

\begin{equation}
    dL(t,T_i)=\int_{T_i}^{T_{i+1}}\xi_i(t,u)dZ^{Q_{T_{i+1}}}(t,u)du
\end{equation}
and

\begin{equation}\label{flibor}
    dL_F(t,T_i)=\int_{T_i}^{T_{i+1}}\xi_i^F(t,u)dZ^{Q^F_{T_{i+1}}}(t,u)du.
\end{equation}

Observe that 

\begin{equation*}
    \xi_i(t,u)=\frac{B(t,T_{i})}{\delta_{i+1} B(t,T_{i+1})}\sigma(t,u),
\end{equation*}

\begin{equation*}
    \xi^F_i(t,u)=\frac{B_F(t,T_{i})}{\delta_{i+1} B_F(t,T_{i+1})}\sigma_F(t,u).
\end{equation*}

Now define the functions $\lambda_i,\lambda_i^F$ such that 

\begin{equation*}
    L(t,T_i)\lambda_i(t,u)=\xi_i(t,u),
\end{equation*}

\begin{equation*}
    L_F(t,T_i)\lambda^F_i(t,u)=\xi^F_i(t,u).
\end{equation*}

Therefore,

\begin{equation}\label{lambdai}
    \lambda_i(t,u)=\frac{1+\delta_{i+1}L(t,T_i)}{\delta_{i+1}L(t,T_i)}\sigma(t,u),
\end{equation}

\begin{equation}\label{lambdaiforeign}
    \lambda^F_i(t,u)=\frac{1+\delta_{i+1}L_F(t,T_i)}{\delta_{i+1}L_F(t,T_i)}\sigma_F(t,u).
\end{equation}

It is clear from \eqref{lambdai} and \eqref{lambdaiforeign} that the functions $\lambda_i$ and $\lambda_i^F$ are stochastic in general. If we assume $\lambda_i$ and $\lambda_i^F$ are deterministic, this leads to the lognormal RFLMM. However, according to equations \eqref{lambdai} and \eqref{lambdaiforeign}, it is clear that the volatility functions $\sigma$ and $\sigma_F$ are stochastic functions in the lognormal RFLMM. We can restate the dynamics of the lognormal LIBOR rates as follows: 

\begin{equation}\label{dlibor}
    dL(t,T_i)=L(t,T_i)\int_{T_i}^{T_{i+1}}\lambda_i(t,u)dZ^{Q_{T_{i+1}}}(t,u)du
\end{equation}
and

\begin{equation}\label{dflibor}
    dL_F(t,T_i)=L_F(t,T_i)\int_{T_i}^{T_{i+1}}\lambda_i^F(t,u)dZ^{Q^F_{T_{i+1}}}(t,u)du.
\end{equation}

In order to express \eqref{flibor} in terms of $Z^{Q_{T_{i+1}}}$, we need to first find a find a no-arbitrage relationship between the random fields $Z^{Q^F_{T_{i+1}}}$ and $Z^{Q_{T_{i+1}}}.$ 

\begin{theorem}

The random fields under the measures $Q_{T_i}$ and $Q^F_{T_i}$ are related by 

\begin{equation}\label{zf}
    dZ^{Q^F_{T_i}}(t,u)=dZ^{Q_{T_i}}(t,u)-\left(\int_t^{T_i}\sigma_{X_i}(t,v)c(u,v)dv\right)dt.
\end{equation}

\end{theorem}

\begin{proof}

Recall that $X(t,T_i)$ is a martingale with respect to $Q_{T_{i}}$. Therefore,

\begin{equation*}
    \frac{1}{X(t,T_i)}=\frac{B(t,T_i)\frac{1}{X(t)}}{B_F(t,T_i)}
\end{equation*}
is a martingale with respect to the foreign $T_i-$forward measure $Q_{T_{i+1}}^F.$

Applying Ito's lemma to $\frac{1}{X(t,T_i)}$ yields that

\begin{equation}\label{1/x}
\begin{split}
    d\left(\frac{1}{X(t,T_i)}\right)&=-\frac{1}{X(t,T_i)^2}dX(t,T_i)+\frac{1}{X(t,T_i)^3}d
    \langle X(t,T_i),X(t,T_i)\rangle\\
    &= \frac{1}{X(t,T_i)}\left[-\int_t^{T_i}\sigma_{X_i}(t,u)dZ^{Q_{T_i}}(t,u)du+\int_{t}^{T_i}\int_{t}^{T_i}\sigma_{X_i}(t,u)\sigma_{X_i}(t,v)c(u,v)dudvdt\right ].
\end{split}
\end{equation}

Observe that the volatility term of $\frac{1}{X(t,T_i)}$, $\sigma_{\frac{1}{X_i}}(t,u)=-\sigma_{X_i}(t,u).$ Since $\frac{1}{X(t,T_i)}$ is a martingale under $Q^F_{T_i},$

\begin{equation}\label{1/x2}
     d\left(\frac{1}{X(t,T_i)}\right)
    = -\frac{1}{X(t,T_i)}\int_t^{T_i}\sigma_{X_i}(t,u)dZ^{Q^F_{T_i}}(t,u)du. 
\end{equation}

Comparing \eqref{1/x} with \eqref{1/x2} proves the desired result.

\end{proof}

As \eqref{zf} provides a way to link the two economies, we can now state the dynamics of the foreign LIBOR rate $L_F$ in terms of the domestic forward measure. 

\begin{theorem}
Fix $i\in \{1,2,...,N-1\}.$ For every $t\in [0,T_i],$


\begin{equation}\label{lfdomesticmeasure}
     \begin{split}
        L_F(T_i,T_i)=L_F(t,T_i)\exp(-\alpha_i(t))&\exp\Bigg(
        \int_t^{T_i}\int_{T_i}^{T_{i+1}}\lambda_i^F(s,u)dZ^{Q_{T_{i+1}}}(s,u)du\\
        &-\frac{1}{2}\int_t^{T_i}\int_{T_i}^{T_{i+1}}\int_{T_i}^{T_{i+1}}\lambda_i^F(s,u)\lambda_i^F(s,v)c(u,v)dudvds\Bigg)
    \end{split}
\end{equation}
where

\begin{equation}\label{alpha}
    \alpha_i(t):=\int_t^{T_i}\int_{T_i}^{T_{i+1}}\int_{s}^{T_{i+1}}\lambda_i^F(s,u)\sigma_{X_{i+1}}(s,v)c(u,v)dvduds.
\end{equation}
\end{theorem}

\begin{proof}
An application of Ito's lemma to $\ln (L_F(t,T_i))$ together with \eqref{zf} yields \eqref{lfdomesticmeasure}.
\end{proof}


Subtracting \eqref{dzq} and \eqref{dzqf} and applying \eqref{zf}  to the resulting equation yields

\begin{equation}
    \int_t^{T_i}\sigma_{X_i}(t,v)c(u,v)dv=\int_{T_i}^{T_{i+1}}(\sigma_F(t,v)-\sigma(t,v))c(u,v)dv+\int_t^{T_{i+1}}\sigma_{X_{i+1}}(t,v)c(u,v)dv
\end{equation}

for all $i=0,1,2,...,N-1$. Hence, recursively we can show that 

\begin{equation}\label{recursive}
    \int_t^{T_i}\sigma_{X_i}(t,v)c(u,v)dv=\sum_{j=i}^{N-1}\int_{T_j}^{T_{j+1}}(\sigma_F(t,v)-\sigma(t,v))c(u,v)dv+\int_t^{T_{N}}\sigma_{X_{N}}(t,v)c(u,v)dv.
\end{equation}

According to \eqref{recursive} together with \eqref{lambdai} and \eqref{lambdaiforeign}, it is not possible to assume that the domestic forward LIBOR rate, the foreign LIBOR rate, and the forward exchange rates are all lognormal simultaneously. This was first observed by Shl\"ogl \cite{schlogl} in the Brownian motion setting. To further elaborate, if we choose $\lambda_i,\lambda_i^F$ to be deterministic for all $i=1,2,\dots,N$, and choose the $\sigma_{X_N}$ also to be deterministic, the rest of the forward exchange rate volatilities are determined by \eqref{recursive}. Since $\sigma$ and $\sigma_F$ are stochastic (when $\lambda_i, \lambda_i^F$ are deterministic), the volatilities of all other maturities of the forward exchange rates are stochastic as they are determined by \eqref{recursive}. We consider the following two cases in this paper: 

\begin{enumerate}[(i)]
    \item The LIBOR rates in all economies are lognormal.
    \item The domestic LIBOR rate and the forward exchange rate are lognormal.
\end{enumerate}

\section{Pricing Quanto Caps in the Cross-Currency RFLMM}

A Quanto is a type of a derivative where the underlying asset is given in terms of one currency and the price of the derivative is quoted in terms of another currency. In this section, we find an approximate valuation formula for a cap written on a foreign LIBOR rate priced in terms of the domestic currency. We find an approximate pricing formula for a Quanto cap assuming lognormal LIBOR rates.\\

To simplify notation, we define the following functions: 

\begin{equation*}
    A(t,T_j)=\frac{\delta_{j+1}L(t,T_j)}{1+\delta_{j+1}L(t,T_j)},
\end{equation*}

\begin{equation*}
    A_F(t,T_j)=\frac{\delta_{j+1}L_F(t,T_j)}{1+\delta_{j+1}L_F(t,T_j)}.
\end{equation*}
By substituting the above functions together with \eqref{lambdai}, and \eqref{lambdaiforeign} into \eqref{recursive} we get

\begin{equation}\label{recursive A}
\begin{split}
    \int_t^{T_i}\sigma_{X_i}(t,v)c(u,v)dv&=\sum_{j=i}^{N-1}A_F(t,T_j)\int_{T_j}^{T_{j+1}}\lambda_j^F(t,u) c(u,v)dv\\&-\sum_{j=i}^{N-1}A(t,T_j)\int_{T_j}^{T_{j+1}}\lambda_j(t,u) c(u,v)dv+\int_t^{T_{N}}\sigma_{X_{N}}(t,v)c(u,v)dv.
\end{split}
\end{equation}

By substituting \eqref{recursive A} into \eqref{alpha} we get 

\begin{equation}
\begin{split}
        \alpha_i(t)&=\int_t^{T_i}\int_{T_i}^{T_{i+1}}\lambda ^F_i(s,u)\sum_{j=i+1}^{N-1}A_F(s,T_j)\int_{T_j}^{T_{j+1}}\lambda_j^F(s,v) c(u,v)dvduds\\
        &-\int_t^{T_i}\int_{T_i}^{T_{i+1}}\lambda^F _i(s,u)\sum_{j=i+1}^{N-1}A(s,T_j)\int_{T_j}^{T_{j+1}}\lambda_j(s,v) c(u,v)dvduds\\
        &+\int_t^{T_i}\int_{T_i}^{T_{i+1}}\int_s^{T_N}\lambda_i^F(s,u)\sigma_{X_N}(s,v)c(u,v)dvduds.
\end{split}
\end{equation}
Since $A$ and $A_F$ are not $\filt_f$ measurable, $\alpha_i(t)$ is not $\filt_t$ measurable. Therefore, we find an approximation $\Tilde{\alpha}_i(t)$ to $\alpha_i(t)$ that is $\filt_t$ measurable by freezing the $L(s,T_i)$ and $L_F(s,T_i)$ at time $t$.

\begin{equation}\label{alphaapprox}
\begin{split}
        \tilde{\alpha}_i(t)&=\sum_{j=i+1}^{N-1}A_F(t,T_j)\int_t^{T_i}\int_{T_i}^{T_{i+1}}\int_{T_j}^{T_{j+1}}\lambda ^F_i(s,u)\lambda_j^F(s,v) c(u,v)dvduds\\
        &-\sum_{j=i+1}^{N-1}A(t,T_j)\int_t^{T_i}\int_{T_i}^{T_{i+1}}\int_{T_j}^{T_{j+1}}\lambda^F _i(s,u) \lambda_j(s,v) c(u,v)dvduds\\
        &+\int_t^{T_i}\int_{T_i}^{T_{i+1}}\int_s^{T_N}\lambda_i^F(s,u)\sigma_{X_N}(s,v)c(u,v)dvduds.
\end{split}
\end{equation}

Freezing approximation methods have been used to approximate valuation formulas for swaps and swaptions in \cite{brigo2007interest} and \cite{wu}.


\begin{theorem}\label{cap}[The Price of a Quanto Caplet]\label{quanto caplet}

Consider a caplet with a reset date $T_i$, settling date $T_{i+1}$, and a strike rate $\kappa$ written on a foreign LIBOR rate $L_F$ with a payoff $\delta_{i+1}\bar{X} (L_F(T_i,T_i)-\kappa)^+$ at time $T_{i+1}$ where $\bar{X}$ is a predetermined fixed exchange rate. The price of the caplet at time $t\leq T_i$ is given by

\begin{equation*}
    QCapl(t,T_i,T_{i+1})\approx\delta_{i+1}\Bar{X}B(t,T_{i+1})[L_F(t,T_i)N(\tilde{d_1}^i(t)-N(\tilde{d_2}^i(t)]
\end{equation*}

where 

\begin{equation*}
\tilde{d_1}^i(t)=    \frac{\ln\left(\frac{L_F(t,T_i)}{k}\right)-\tilde{\alpha}_i(t)+\frac{1}{2}\tilde{\Omega}_i(t)}{\sqrt{\tilde{\Omega}_i(t)}},
\end{equation*}

\begin{equation*}
\tilde{d_2}^i(t)=\frac{\ln\left(\frac{L_F(t,T_i)}{k}\right)-\tilde{\alpha}_i(t)-\frac{1}{2}\tilde{\Omega}_i(t)}{\sqrt{\tilde{\Omega}_i(t)}},
\end{equation*}
and

\begin{equation*}
\tilde{\Omega}_i(t)=\int_t^{T_i}\int_{T_i}^{T_{i+1}}\int_{T_i}^{T_{i+1}}\lambda_i^F(s,u)\lambda_i^F(s,v)c(u,v)dudvds.
\end{equation*}
\end{theorem}

\begin{proof}
The price of this caplet at a time $t\in [0,T_i]$ with respect to the $Q_{T_{i+1}}$ forward measure (\cite{musiela} corollary 9.6.1) can be expressed as follows:

\begin{equation*}
    Qcapl(t,T_i,T_{i+1})=\delta_{i+1}\Bar{X}B(t,T_{i+1})E_t^{Q_{T_{i+1}}}[(L_F(T_i,T_i)-k)^+]=\delta_{i+1}\Bar{X}B(t,T_{i+1})[I_1-kI_2]
\end{equation*}
where 
\begin{equation*}
    I_1=E_t^{Q_{T_{i+1}}}[L_F(T_i,T_i)\mathbb{I}_D],
\end{equation*}

\begin{equation*}
    I_2=E_t^{Q_{T_{i+1}}}[\mathbb{I}_D]
\end{equation*}
and $D=\{L_F(T_i,T_i)>k\}$ and $E_t[\cdot]$ denotes the conditional expectation with respect to $\filt_t.$

By taking the natural log of both sides of \eqref{lfdomesticmeasure} and replacing $\alpha_i(t)$ with $\tilde{\alpha}_i(t)$ we get that 

\begin{equation}\label{logfl2}
\begin{split}
    \ln L_F(T_i,T_i)&\approx\ln L_F(t,T_i)-\tilde{\alpha}_i(t)+\int_t^{T_i}\int_{T_i}^{T_{i+1}}\lambda^F_i(s,u)dZ^{Q_{T_{i+1}}}(s,u)du\\
        &-\frac{1}{2}\int_t^{T_i}\int_{T_i}^{T_{i+1}}\int_{T_i}^{T_{i+1}}\lambda^F_i(s,u)\lambda^F_i(s,v)c(u,v)dudvds.
\end{split}
\end{equation}

Therefore, the $Q_{T_{i+1}}$ conditional variance and expectation of $\ln L_F(T_i,T_i)$ are as follows:

\begin{equation}\label{logv}
\begin{split}
    Var_t^{Q_{T_{i+1}}}(\ln L_F(T_i,T_i))&\approx\int_t^{T_i}\int_{T_i}^{T_{i+1}}\int_{T_i}^{T_{i+1}}\lambda^F_i(s,u)\lambda^F_i(s,v)c(u,v)dudvds\\
    &:=\tilde{\Omega}_i(t),
\end{split}    
\end{equation}

\begin{equation}\label{loge}
    E_t^{Q_{T_{i+1}}}(\ln L_F(T_i,T_i))\approx\ln L_F(t,T_i)-\tilde{\alpha}_i(t)-\frac{1}{2}\tilde{\Omega}_i(t).
\end{equation}

Since $\ln L_F(T_i,T_i)\mid \filt_t$ is approximately normally distributed,  

\begin{equation}
    I_2\approx N\left(\frac{\ln\left(\frac{L_F(t,T_i)}{k}\right)-\tilde{\alpha}_i(t)-\frac{1}{2}\tilde{\Omega}_i(t)}{\sqrt{\tilde{\Omega}_i(t)}}\right)
\end{equation}

where $N(\cdot)$ denotes the cumulative distribution function of a standard normal distribution. To evaluate $I_1$ we first define an equivalent measure $\hat{Q}_{T_{i+1}}\sim Q_{T_{i+1}}$ by

    \begin{equation*}
    \frac{d\hat{Q}_{T_{i+1}}}{dQ_{T_{i+1}}}=\eta(T_i,T_i,T_{i+1})
\end{equation*}

where 

\begin{equation*}
    \eta(t,T_i,T_{i+1}):=\exp\left(\int_0^t\int_{T_i}^{T_{i+1}} \lambda_i^F(s,u)dZ^{Q_T}(s,u)du
    -\frac{1}{2}\int_0^t\int_{T_i}^{T_{i+1}}\int_{T_i}^{T_{i+1}} \lambda_i^F(s,u)\lambda_i^F(s,v)c(u,v)dudvds\right)
\end{equation*}

is a positive $Q_{T_{i+1}}-$martingale. The process $Z^{\hat{Q}_{T_{i+1}}}$ defined by

\begin{equation}\label{dzhat}
    dZ^{\hat{Q}_{T_{i+1}}}(s,u)= dZ^{{Q}_{T_{i+1}}}(s,u)-\left(\int_{T_i}^{T_{i+1}}\lambda_i^F(s,v)c(u,v)dv\right)ds
\end{equation}
is a  $\hat{Q}_{T_{i+1}}$ random field based on Girsanov theorem (\cite{oks} theorem 8.6.4). Under  $\hat{Q}_{T_{i+1}}$, \eqref{logfl2} transforms to 

\begin{equation}\label{logfltrans}
\begin{split}
    \ln L_F(T_i,T_i)&\approx\ln L_F(t,T_i)-\tilde{\alpha}_i(t)+\int_t^{T_i}\int_{T_i}^{T_{i+1}}\lambda^F_i(s,u)(s,u)dZ^{Q_{T_{i+1}}}(s,u)du\\
        &+\frac{1}{2}\int_t^{T_i}\int_{T_i}^{T_{i+1}}\int_{T_i}^{T_{i+1}}\lambda^F_i(s,u)\lambda^F_i(s,v)c(u,v)dudvds.
\end{split}
\end{equation}
This implies that

\begin{equation}\label{logehat}
    E_t^{\hat{Q}_{T_{i+1}}}[\ln L_F(T_i,T_i)]\approx\ln L_F(t,T_i)-\tilde{\alpha}_i(t)+\frac{1}{2}\tilde{\Omega}_i(t).
\end{equation}

Notice that the conditional variance remains unchanged under the change of measure. Using Bayes' rule (\cite{oks} lemma 8.6.2) we evaluate $I_1$ as follows:







\begin{equation*}
\begin{split}
    I_1&=E^{Q_{T_{i+1}}}_t[L(T_i,T_i)\mathbb{I}_D]\\
    &= E^{Q_{T_{i+1}}}_t[L(t,T_i)\eta(T_i,T_i,T_{i+1})\eta(t,T_i,T_{i+1})^{-1}\mathbb{I}_D]\\
    &=L(t,T_i)\frac{E^{Q_{T_{i+1}}}_t[\eta(T_i,T_i,T_{i+1})\mathbb{I}_D]}{E^{Q_{T_{i+1}}}_t[\eta(T_i,T_i,T_{i+1})]}\\
    &=L(t,T_i)E^{\hat{Q}_{T_{i+1}}}_t[\mathbb{I}_D]\\
    &\approx L(t,T_i)N\left(\frac{\ln\frac{L_F(t,T_i)}{k}-\tilde{\alpha}_i(t)+\frac{1}{2}\tilde{\Omega}_i(t)}{\sqrt{\tilde{\Omega}_i(t)}}\right).
\end{split}
\end{equation*}
\end{proof}
The next result follows immediately from the preceding theorem.
\begin{corollary} (The price of a Quanto Cap)

The price of a Quanto cap with reset dates $T_{i}$ for $i=0,1,2,...N-1$ and settling dates $T_i$ for $i=1,2,...N$ with a strike rate $\kappa$ at time $t\leq T_0$ equals

\begin{equation*}
    QCap(t)=\sum_{i=0}^{N-1}QCapl(t,T_i,T_{i+1}).
\end{equation*}

\end{corollary}

Let us now see how an approximate pricing formula for a Quanto cap can be derived when the $L_F$ is not lognormal. From \eqref{lambdaiforeign},

\begin{equation}\label{lamif}
    \lambda_i^F(t,u)=\frac{\sigma_F(t,u)}{A_F(t,T_i)}. 
\end{equation}
If we allow $\lambda_i^F$ to be stochastic, it is now possible for us to assume that $\sigma_F$ is deterministic. Applying \eqref{lamif} to \eqref{dflibor} yields

\begin{equation}\label{dfliborstoch}
dL(t,T_i)=\frac{L_F(t,T_i)}{A_F(t,T_i)}\int_{T_i}^{T_{i+1}}\sigma_F(t,u)dZ^{Q^F_{T_{i+1}}}(t,u)du
\end{equation}

Using \eqref{zf} to \eqref{dfliborstoch} we get

\begin{equation}
\begin{split}
        \ln L_F(T_i,T_i)=&\ln L_F(t,T_i)+\int_t^{T_{i+1}}\int_{T_i}^{T_{i+1}}\frac{\sigma_F(s,u)}{A_F(s,T_i)}dZ^{Q_{T_{i+1}}}(t,u)du\\&-\int_t^{T_{i+1}}\int_{T_i}^{T_{i+1}}\int_{T_i}^{T_{i+1}}\frac{\sigma_F(s,u)\sigma_{X_{i+1}}(s,v)c(u,v)}{A_F(s,T_i)^2}dvduds\\
        &-\int_t^{T_{i+1}}\int_{T_i}^{T_{i+1}}\int_{T_i}^{T_{i+1}}\frac{\sigma_F(s,u)\sigma_{F}(s,v)c(u,v)}{A_F(s,T_i)^2}dvduds. 
\end{split}
\end{equation}

By freezing $A_F$ at time $t$, we get the following approximation:

\begin{equation}\label{beta}
\begin{split}
        \ln L_F(T_i,T_i)\approx&\ln L_F(t,T_i)+\frac{1}{A_F(t,T_i)}\int_t^{T_{i+1}}\int_{T_i}^{T_{i+1}}\sigma_F(s,u)dZ^{Q_{T_{i+1}}}(t,u)du\\&-
\underbrace{\frac{1}{A_F(t,T_i)^2}\int_t^{T_{i+1}}\int_{T_i}^{T_{i+1}}\int_{T_i}^{T_{i+1}}\sigma_F(s,u)\sigma_{X_{i+1}}(s,v)c(u,v)dvduds}_{=:\beta_i(t)}
\\
        &-\frac{1}{2A_F(t,T_i)^2}\int_t^{T_{i+1}}\int_{T_i}^{T_{i+1}}\int_{T_i}^{T_{i+1}}\sigma_F(s,u)\sigma_{F}(s,v)c(u,v)dvduds. 
\end{split}
\end{equation}
Therefore, 

\begin{equation}\label{logbeta}
    \begin{split}
        L_F(T_i,T_i)\approx&L(t,T_i)\exp\left(-\beta_i(t)\right)\exp\Big(\frac{1}{A_F(t,T_i)}\int_t^{T_{i+1}}\int_{T_i}^{T_{i+1}}\sigma_F(s,u)dZ^{Q_{T_{i+1}}}(t,u)du\\
        &-\frac{1}{2A_F(t,T_i)^2}\int_t^{T_{i+1}}\int_{T_i}^{T_{i+1}}\int_{T_i}^{T_{i+1}}\sigma_F(s,u)\sigma_{F}(s,v)c(u,v)dvduds\Bigg).
    \end{split}
\end{equation}

Setting $D=\{L_F(T_i,T_i)>k\}$, observe that

\begin{equation}
    E_t^{Q_{T_{i+1}}}[L_F(T_i,T_i)\mathbb{I}_D]=L(t,T_i)e^{-\beta_i(t)} E_t^{\hat{Q}_{T_{i+1}}}[\mathbb{I}_D]
\end{equation}
where $\hat{Q}_{T_{i+1}}$ is defined by the relationship 

\begin{equation}
    dZ^{\hat{Q}_{T_{i+1}}}(s,u)=dZ^{{Q}_{T_{i+1}}}(s,u)-\frac{1}{A_F(t,T_i)}\int_{T_i}^{T_{i+1}}\sigma_F(s,v)c(u,v)dvds.
\end{equation}

Similar to the proof of Theorem \ref{quanto caplet}, it can be shown that 

\begin{equation*}
    E_t^{\hat{Q}_{T_{i+1}}}[\mathbb{I}_D]=N\left(\frac{\ln\left(\frac{L_F(T_i,T_i)}{k}\right) -\beta_i(t)+\frac{1}{2}\gamma_i(t)}{\sqrt{\gamma_i(t)}}\right)
\end{equation*}
and 

\begin{equation*}
     E_t^{{Q}_{T_{i+1}}}[\mathbb{I}_D]=N\left(\frac{\ln\left(\frac{L_F(T_i,T_i)}{k}\right) -\beta_i(t)-\frac{1}{2}\gamma_i(t)}{\sqrt{\gamma_i(t)}}\right)
\end{equation*}
where 

\begin{equation*}
    \gamma_i(t):=\frac{1}{A_F(t,T_i)^2}\int_t^{T_{i+1}}\int_{T_i}^{T_{i+1}}\int_{T_i}^{T_{i+1}}\sigma_F(s,u)\sigma_{F}(s,v)c(u,v)dvduds.
\end{equation*}

The following theorem states the price of a Quanto cap when $L_F$ is not lognormal. 
\begin{theorem} 

The price of a Quanto cap ( when $L_F$ is not lognormal) with reset dates $T_{i}$ for $i=0,1,2,...N-1$ and settling dates $T_i$ for $i=1,2,...N$ with a strike rate $\kappa$ at time $t\leq T_0$ equals

\begin{equation*}
    QCap(t)\approx\sum_{i=0}^{N-1}B(t,T_{i+1})(L_F(t,T_i)e^{-\beta_i(t)}N(d_1^i)-kN(d_2^i))
\end{equation*}

where 

\begin{equation*}
    d_1^i:=\frac{\ln\left(\frac{L_F(T_i,T_i)}{k}\right) -\beta_i(t)+\frac{1}{2}\gamma_i(t)}{\sqrt{\gamma_i(t)}}
\end{equation*}
 and 
 \begin{equation*}
    d_2^i:=\frac{\ln\left(\frac{L_F(T_i,T_i)}{k}\right) -\beta_i(t)-\frac{1}{2}\gamma_i(t)}{\sqrt{\gamma_i(t)}}.
\end{equation*}

\end{theorem}

\section{Cross-Currency Swaps}

In this section we consider the pricing of cross-currency swaps. A cross-currency swap is a type of an interest rate swap that involves at least one foreign interest rate. A discussion of pricing cross-currency swaps can be found in Brownian motion setting can be found in \cite{musiela}. Consider a float-to-float swap paid in arrears with reset dates $T_0,T_1,\dots,T_{N-1}$ and payment dates $T_1,T_2,\dots,T_N$ and some principal amount $N$ (stated in the domestic currency). At a payment date $T_{i+1}$, the buyer of the swap pays the seller an amount $NL(T_i,T_i)\delta_{i+1}$ and receives $NL_F(T_i,T_i)\delta_{i+1}$ from the seller. Both payments are made in the domestic currency. Therefore, the buyer is not exposed to exchange-rate risk. For simplicity, if we set $N=1$. Then the net payment at time $T_{i+1 }$ is $(L_F(T_i,T_i)-L(T_i,T_i))\delta_{i+1}$. Therefore, the value at time $t\leq T_0$ is

\begin{equation}
    CCS(t)=\sum_{i=0}^{N-1}\delta_{i+1}B(t,T_{i+1})E^{Q_{T_{i+1}}}_t\left[(L_F(T_i,T_i)-L(T_i,T_i))\right]. 
\end{equation}

\begin{theorem}

The price of a float-to float cross-currency swap with reset dates $T_0,T_1,\dots,T_{N-1}$ and payment dates when 

\begin{enumerate}[(i)]
    \item $L$ and $L_F$ are lognormal is:
     \begin{equation}
        CCS(t)\approx\sum_{i=0}^{N-1}\delta_{i+1}B(t,T_{i+1})(L_F(t,T_i)e^{-\tilde{\alpha}_i(t)}-L(t,T_i)),    
    \end{equation}
    
    \item $L$ and $X$ are lognormal is:
    \begin{equation}
        CCS(t)\approx\sum_{i=0}^{N-1}\delta_{i+1}B(t,T_{i+1})(L_F(t,T_i)e^{-\tilde{\beta}_i(t)}-L(t,T_i)),
    \end{equation}
\end{enumerate}
where $\tilde{\alpha}_i(t)$ and $\beta_i(t)$ and defined in \eqref{alphaapprox} and \eqref{beta} respectively.

\end{theorem}

\begin{proof}
The proof directly follows from equations \eqref{logfl2}, \eqref{logbeta}, and the fact that $E^{Q_{T_{i+1}}}_t[L(T_i,T_i)]=L(t,T_i).$
\end{proof}

\section{Exchange Rate Options}

In this section we focus on deriving an exact closed-form pricing formula for an option written on the spot exchange rate $X(t)$ assuming that the forward exchange rate is lognormal. Consider a call option written on the spot exchange rate $X$ with an expiration date $T$, strike rate $K$, and principal amount of 1 unit of domestic currency. The payoff at time $T$ can be written as 

\begin{equation*}
    C_T^X= (X(T)-K)^+= (X(T,T)-K)^+.
\end{equation*}

\begin{theorem}\label{exchangeoption}Consider a call option written on the spot exchange rate $X$ which expires at time $T$ and strike rate $k.$ Then the no-arbitrage price of this call option at time $t\leq T$ denoted by $C_t^X$ is given by 

\begin{equation}
    C_t^X=B(t,T)(X(t,T)N(d_1)-kN(d_2))
\end{equation}

where 

\begin{equation*}
    d_1:=\frac{\ln\left(\frac{X(t,T)}{k}\right)+\frac{1}{2}\gamma(t,T)}{\sqrt{\gamma(t,T)}},
\end{equation*}

\begin{equation*}
    d_2:= \frac{\ln\left(\frac{X(t,T)}{k}\right)-\frac{1}{2}\gamma(t,T)}{\sqrt{\gamma(t,T)}},
\end{equation*}
and 

\begin{equation*}
    \gamma(t,T):=\int_t^T\int_s^T\int_s^T \sigma_{X_T}(s,u)\sigma_{X_T}(s,v)c(u,v)dudvds.
\end{equation*}

\end{theorem}

\begin{proof}

Recall that 

\begin{equation*}
    X(t,T):=\frac{B_F(t,T) X(t)}{B(t,T)}
\end{equation*}

and the dynamics of $X(t,T)$ is given by 

\begin{equation}
    dX(t,T)=X(t,T)\int_t^{T}\sigma_{X_T}(t,u)dZ^{Q_{T}}(t,u)du. 
\end{equation}

By Ito's rule

\begin{equation*}
    d\ln X(t,T)= \int_t^T \sigma_{X_T}(t,u)dZ^{Q_T} (t,u) du-\frac{1}{2}\int_t^T \int_t^T \sigma_{X_T}(t,u)\sigma_{X_T}(t,v)c(u,v)dudvdt. 
\end{equation*}

Integrating the above equation yields

\begin{equation}\label{Xdyanamics}
\begin{split}
 X(T,T)=X(t,T)\exp\Bigg( &\int_t^T\int_s^T \sigma_{X_T}(s,u)dZ^{Q_T}(s,u)du\\
 &-\frac{1}{2}\int_t^T \int_s^T\int_s^T \sigma_{X_T}(s,u)\sigma_{X_T}(s,v)c(u,v)dudvds\Bigg)
\end{split}
\end{equation}

 The price with respect to the forward measure $Q_T$ is 

\begin{equation*}
\begin{split}
        C_t^X&=B(t,T)E_t^{Q_T}\left[(X(T,T)-k)^+\right]\\
        &= B(t,T)[I_1-kI_2]
\end{split}
\end{equation*}
where 

\begin{equation*}
    I_1:=E_t^{Q_T}[X(T,T)\mathbb{I}_D]
\end{equation*}

\begin{equation*}
    I_2:=E_t^{Q_T}[\mathbb{I}_D]
\end{equation*}

and $D=\{X(T,T)>k\}.$ Let us first evaluate $I_2$. For this, we have to find the mean and the variance of $\ln X(T,T)$ with respect to $Q_T$ conditioned on $\filt_t.$ By \eqref{Xdyanamics}, 

\begin{equation}
\begin{split}
    \ln X(T,T)&= \ln X(t,T) +\int_t^T\int_s^T \sigma_{X_T}(s,u)dZ^{Q_T}(s,u)du\\
    &-\frac{1}{2}\int_t^T\int_s^T\int_s^T \sigma_{X_T}(s,u)\sigma_{X_T}(s,v)c(u,v)dudvds.
\end{split}
\end{equation}

Therefore,

\begin{equation}\label{lnxvar}
\begin{split}
    Var_t^{Q_T}[\ln X(T,T) ]&=\int_t^T\int_s^T\int_s^T \sigma_{X_T}(s,u)\sigma_{X_T}(s,v)c(u,v)dudvds\\
    &:=\gamma(t,T), 
    \end{split}
\end{equation}

and 

\begin{equation}\label{lnxexp}
    \begin{split}
        E^{Q_T}_t[\ln X(T,T)]&=\ln X(t,T)-\frac{1}{2}\int_t^T\int_s^T\int_s^T \sigma_{X_T}(s,u)\sigma_{X_T}(s,v)c(u,v)dudvds\\
        &=\ln X(t,T)-\frac{1}{2}\gamma(t,T).
    \end{split}
\end{equation}

Since $\ln X(T,T)\mid \filt_t$ is normally distributed,

\begin{equation*}
\begin{split}
    I_2&=N\left(\frac{\ln\left(\frac{X(t,T)}{k}\right)-\frac{1}{2}\gamma(t,T)}{\sqrt{\gamma(t,T)}}\right).
\end{split}
\end{equation*}

To evaluate $I_1,$ we define a new random field ${Z}^{\hat{Q}_T}$ by

\begin{equation}\label{zhat}
    d\hat{Z}^{Q_T}(s,u)=dZ^{Q_T}(s,u)-\left(\int_s^T\sigma_X(s,v)c(u,v)dv\right)ds.
\end{equation}

Then by Girsanov theorem, there exists a measure $\hat{Q}_T\sim Q_T$ such that $\{Z(t,T)\}_t$ is a $\hat{Q}_T$ Brownian motion. In particular, the Radon-Nikodym derivative of $\hat{Q}_T$ with respect to $Q_T$ is given by 

\begin{equation*}
    \frac{d\hat{Q}_T}{dQ_T}=\eta(T,T)
\end{equation*}

where 

\begin{equation*}
    \eta(t,T)=\exp\left(\int_0^t\int_s^T \sigma_{X_T}(s,u)dZ^T(s,u)du
    -\frac{1}{2}\int_0^t\int_s^T\int_s^T \sigma_{X_T}(s,u)\sigma_{X_T}(s,v)c(u,v)dudvds\right)
\end{equation*}

is a strictly positive martingale with respect to $Q_T$. In particular, $\eta(t,T)=E^{Q_T}_t[\eta(T,T)]$ for all $t\leq T.$
Applying \eqref{zhat} to \eqref{Xdyanamics} we can express the dynamics of $X(T,T)$ under the measure $\hat{Q}_T$ as 
\begin{equation}
    \begin{split}
         X(T,T)=X(t,T)\exp\Bigg( &\int_t^T\int_s^T \sigma_X(s,u)d\hat{Z}^{Q_T}\\
 &+\frac{1}{2}\int_t^T \int_s^T \sigma_X(s,u)\sigma_X(s,v)c(u,v)dudvds\Bigg).
    \end{split}
\end{equation}

Therefore, using Bayes' rule we show that 

\begin{equation*}
\begin{split}
    I_1&=E^{Q_T}_t[X(T,T)\mathbb{I}_D]\\
    &= E^{Q_T}_t[X(t,T)\eta(T,T)\eta(t,T)^{-1}\mathbb{I}_D]\\
    &=X(t,T)\frac{E^{Q_T}_t[\eta(T,T)\mathbb{I}_D]}{E^{Q_T}_t[\eta(t,T)]}\\
    &=X(t,T)E^{\hat{Q}_T}_t[\mathbb{I}_D].
\end{split}
\end{equation*}

To evaluate $E^{\hat{Q}_T}_t[\mathbb{I}_D]$, observe that

\begin{equation}\label{lnxexp}
    \begin{split}
        E^{\hat{Q}_T}_t[\ln X(T,T)\mid \filt_t]&=\ln X(t,T)+\frac{1}{2}\int_t^T\int_s^T\int_s^T \sigma_X(s,u)\sigma_X(s,v)c(u,v)dudvds\\
        &=\ln X(t,T)+\frac{1}{2}\gamma(t,T).
    \end{split}
\end{equation}

Therefore,

\begin{equation*}
    I_1=X(t,T)N\left(\frac{\ln\left(\frac{X(t,T)}{k}\right)+\frac{1}{2}\gamma(t,T)}{\sqrt{\gamma(t,T)}}\right).
\end{equation*}

\end{proof}

\printbibliography
\end{document}